\documentclass[11pt,twoside]{amsart}
\usepackage{amssymb,amsmath,amsfonts,amsthm}
\usepackage[T1]{fontenc}
\usepackage[english]{babel}
\usepackage{enumerate}
\theoremstyle{plain}
\newtheorem{theorem}{Theorem}[section]
\newtheorem{lemma}[theorem]{Lemma}
\newtheorem{proposition}[theorem]{Proposition}
\newtheorem{corollary}[theorem]{Corollary}
\theoremstyle{remark}
\newtheorem{remark}[theorem]{Remark}
\theoremstyle{definition}
\newtheorem{assumption}{Assumption}

\newtheorem{condition}[assumption]{Condition}

\newtheorem{definition}[theorem]{Definition}
\newtheorem{example}[theorem]{Example}
\newcommand{\ZZ}{\mathbb{Z}}
\newcommand{\QQ}{\mathbb{Q}}
\newcommand{\RR}{\mathbb{R}}
\newcommand{\CC}{\mathbb{C}}
\newcommand{\NN}{\mathbb{N}}
\newcommand{\EE}{\mathbb{E}}
\newcommand{\PP}{\ensuremath{\mathbb{P}}}
\newcommand{\drm}{\ensuremath{\mathrm{d}}}
\newcommand{\cA}{\mathcal{A}}
\newcommand{\Eins}{\mathbf{1}}
\newcommand{\cF}{\mathcal{F}}
\newcommand{\cB}{\mathcal{B}}
\newcommand{\Pro}{\ensuremath{p}}
\newcommand{\Inc}{\ensuremath{\iota}}
\DeclareMathOperator{\cov}{cov}
\DeclareMathOperator{\supp}{\operatorname{supp}}
\DeclareMathOperator*{\esssup}{ess\,sup}
\renewcommand{\epsilon}{\varepsilon}
\usepackage[BCOR=15mm,DIV=11]{typearea}
\usepackage{microtype}
\title[Discrete alloy-type models: Regularity of distributions]{Discrete alloy-type models: Regularity of distributions and recent results}
\author{Martin Tautenhahn}
\author{Ivan Veseli\'c}
\address{Technische Universit\"at Chemnitz, Fakult\"at f\"ur Mathematik, D-09107 Chemnitz}
\date{}
\begin{document}
\begin{abstract}\noindent
We consider discrete random Schr\"odinger operators on $\ell^2 (\ZZ^d)$ with a potential of discrete alloy-type structure. 
That is, the potential at lattice site $x \in \ZZ^d$ is given by a linear combination of independent identically distributed random variables, possibly
 with sign-changing coefficients. In a first part we show that the discrete alloy-type model is not uniformly $\tau$-H\"older continuous, a frequently used condition in the literature of Anderson-type models with general random potentials. In a second part we review recent results on regularity properties of spectral data and localization properties for the discrete alloy-type model.
\end{abstract}
\maketitle
\section{Introduction}
We consider discrete Schr\"odinger operators on $\ell^2(\ZZ^d)$, where $\ZZ^d$ is the $d$-dimensional integer lattice.
The potential of the discrete Schr\"odinger operator is given by a stochastic field. 
Thus we are dealing with generalizations of the standard Anderson model.
We are interested in properties of distributions of spectral data of such random operators, as well as of their restrictions to finite cubes $\Lambda \subset \ZZ^d$. 
An appropriate control of these distributions allows one to conclude almost sure spectral and dynamical localization 
for the random Schr\"odinger operator on $\ell^2(\ZZ^d)$.
More precisely, we discuss in the paper the following issues:
\begin{itemize}
 \item 
We review results on discrete alloy-type models proven in \cite{ElgartTV-10,ElgartTV-11} and \cite{PeyerimhoffTV-11,LeonhardtPTV-13}.
They concern localization criteria based on the multiscale analysis and the fractional moment method and can be considered as generalizations of earlier results obtained in \cite{Veselic-10a,Veselic-10b,TautenhahnV-10}.
Also, we discuss related results in the recent paper \cite{ElgartSS-12} and highlight the role of a reverse H\"older inequality 
in the argument of \cite{ElgartSS-12}.
\item
We present Minami estimates and Poisson statistics of eigenvalues proven in \cite{TautenhahnV-13}
for a class of discrete alloy-type models. This extend results of \cite{Minami-96} beyond the Anderson model.
\item
Alloy-type potentials are a specific type of a correlated stochastic field.
While there are abstract localization results in the literature concerning 
correlated random potentials, they rarely cover those of alloy-type.
We show this by checking the relevant regularity properties of the conditional distributions of the stochastic field. 
\item
This prompts a careful consideration of conditional distributions. 
In the literature on random Schr\"odinger operators these are sometimes not treated correctly.
We show how to deal with certain measurability issues and give a (counter)example, 
which shows how badly conditional distributions may behave,  even for innocently looking alloy-type potentials.
\end{itemize}
Let us put these statements into context. 
The Anderson model is a Schr\"odinger operators on $\ell^2(\ZZ^d)$ with potential given by an independent identically distributed (i.i.d.) sequence of random variables
$V_\cdot(x), x\in \ZZ^d$. One expects that for energies near the infimum and supremum of the spectrum, as well as at large disorder, 
this model exhibits localization, i.e.~discrete spectrum with exponentially decaying  eigenfunctions, almost surely.
However, proofs of this claim depend on regularity conditions on the distribution of random variables. 
For instance, for the Anderson model with Bernoulli distributed variables so far localization was proven only in one space dimension.
If we consider more general random potentials, where each random variable $V_\cdot(x), x\in \ZZ^d$, has the same marginal distribution $\nu$, 
but they are no longer independent, additional distinctions are necessary.
Now one has to impose regularity conditions on the finite-dimensional distributions of the process  $V_\cdot(x), x\in \ZZ^d$, 
i.e. the joint distribution of a finite subcollection of random variables.
Alternatively, one can formulate regularity hypotheses on the conditional distributions.
While the ultimate goal is to formulate regularity hypotheses which can be  used to derive localization, 
an intermediate step is to derive regularity of spectral data. More precisely, 
one wants to show, that if the distribution of the stochastic process  $V_\cdot(x), x\in \ZZ^d$, is sufficiently regular, 
then the distribution of spectral data is so as well. Abstractly speaking: the pushforward map preserves the regularity of probability measures.
Let us give an illustration. If the distribution function of $\nu$ is Lipschitz continuous, then the integrated density of states inherits this property.
This is called a Wegner estimate.
\par
We review here a number of positive results in this direction for correlated fields  $V_\cdot(x), x\in \ZZ^d$, which arise
as an alloy-type potential. In stochastic data analysis such models are known as (multidimensional) moving average processes.
The mentioned results include Wegner estimates, uniform bounds on fractional moments and exponential decay of fractional moments.
This are results obtained in \cite{ElgartTV-11,PeyerimhoffTV-11,LeonhardtPTV-13}.
They have been extended in the recent papers \cite{Krueger-12} and \cite{ElgartSS-12}. The latter one will be discussed in Section~\ref{s:ESS}.
In particular, for the case of alloy-type potentials with large disorder we give a short and direct modification of the proof of  
a subharmonicity inequality crucial for the fractional moment method. Furthermore, we single out a reverse H\"older-inequality
as the pivot estimate in the strategy of \cite{ElgartSS-12}. 
\par
While the Wegner estimate concerns a bound on the probability of finding an eigenvalue at all in an energy interval,
the Minami estimate bounds the the probability of finding at least two eigenvalue in an energy interval.
For a specific class of discrete alloy-type models Minami's result \cite{Minami-96} has been generalized in \cite{TautenhahnV-13}.
We discuss this in Section~\ref{s:regularity}. The implications for the asymptotic statistics of eigenvalues is presented in Section~\ref{s:physical_implications}.
\par
The papers \cite{DreifusK-91,AizenmanM-93,AizenmanG-98,Hundertmark-00,AizenmanSFH-01,Hundertmark-08} give abstract regularity conditions, formulated in terms of conditional distributions, 
which ensure localization for discrete Schr\"odinger operators with random potential.
We show that these regularity conditions are not satisfied for alloy-type potentials with bounded values, see Section~\ref{s:regularity_pot} for a precise statement. A very interesting borderline behavior is encountered for alloy-type potentials with Gaussian coupling constants. In this case the above mentioned regularity conditions may be satisfied or not, depending on the specific choice of the single-site potential.
In Section~\ref{s:regularity_pot} we also show how to define carefully the associated concentration function or modulus of continuity. 
This concerns the measurability of a supremum over an uncountable set. Such measurability issues are encountered in other areas of probability theory, 
for instance in the context of the Glivenko-Cantelli Theorem or the definition of Markov transition kernels. 
Using regular versions of the conditional expectation we show in detail how to define the concentration functions rigorously.
The reason to devote so much attention to this topic is, that in the literature on random Schr\"odinger operators conditional distributions are not always 
treated correctly. 
There is the misconception that a moving average process inherits the regularity properties of the i.i.d random variables on which it is based.
Section~\ref{s:regularity_pot} shows that if one starts with regularly distributed i.i.d.~random variables and uses them to define a discrete alloy-type potential 
(or moving average process) the resulting conditional distributions are quite singular. These results have been formulated before in the technical reports \cite{TautenhahnV-10b,TautenhahnV-13b}.

To summarize, multidimensional Anderson models without independence condition are still not very well understood.
Exceptions are Gaussian processes treated rigorously in \cite{DreifusK-91,AizenmanSFH-01} and discrete alloy-type potentials
treated in the above mentioned papers.
\section{Notation and model}
\subsection{General random Schr\"odinger operators on $\ZZ^d$}
Let $(\Omega , \mathcal{A} , \PP)$ be a probability space and $\eta_k : (\Omega , \mathcal{A}) \to (\RR , \mathcal{B} (\RR))$, $k \in \ZZ^d$, be real-valued random variables. We define the product space $Z = \times_{k \in \ZZ^d} \RR$ equipped with the product $\sigma$-algebra $\mathcal{Z} = \otimes_{k \in \ZZ^d} \mathcal{B} (\RR)$. The collection $(\eta_k)_{k \in \ZZ^d}$ will be denoted by 
\[
\eta := (\eta_k)_{k \in \ZZ^d} : (\Omega , \mathcal{A}) \to (Z , \mathcal{Z}) .
\]
The expectation with respect to the probability measure $\PP$ will be denoted by $\EE$.
A discrete random Schr\"odinger operator is given by a family of self-adjoint operators
\begin{equation}\label{eq:hamiltonian}
H_\omega = -\Delta + \lambda V_\omega, \quad \omega \in \Omega ,
\end{equation}
on $\ell^2 (\ZZ^d)$. Here $\lambda > 0$ measures the strength of the disorder present in the model, $\Delta$ denotes the discrete Laplace operator and $V_\omega$ is a multiplication operator. They are defined by
\[
 (\Delta \psi) (x) = \sum_{\lvert y-x \rvert_1} \psi (y), \quad \text{and} \quad (V_\omega \psi)(x) = \eta_x(\omega) \psi (x) .
\]
We assume that $H_\omega$ is for each $\omega \in \Omega$ a self-adjoint operator (on some dense domain $D_\omega \subset \ell^2 (\ZZ^d)$). This is for example satisfied if the random potentials $\eta_k$, $k \in \ZZ^d$ are uniformly bounded random variables. If the potential values are not uniformly bounded, we recall that $H_\omega$ is essentially self-adjoint on the set of compactly supported functions, see e.g.\ \cite{Kirsch-08}. 

For the operator $H_\omega$ in \eqref{eq:hamiltonian} and $z \in \CC \setminus \sigma
(H_\omega)$ we define the corresponding \emph{resolvent} by $G_\omega (z)
= (H_\omega - z)^{-1}$. For the \emph{Green function}, which assigns
to each  $(x,y) \in \ZZ^d \times \ZZ^d$ the corresponding matrix element of the
resolvent, we use the notation
\begin{equation*} \label{eq:greens}
G_\omega (z;x,y) := \langle \delta_x , (H_\omega - z)^{-1}\delta_y \rangle.
\end{equation*}
For $\Gamma \subset \ZZ^d$, $\delta_k \in \ell^2 (\Gamma)$ denotes the
Dirac function given by $\delta_k (k) = 1$ for $k \in \Gamma$ and
$\delta_k (j) = 0$ for $j \in \Gamma \setminus \{k\}$.
Let $\Gamma \subset \ZZ^d$. We define the canonical restriction $\Pro_{\Gamma} : \ell^2 (\ZZ^d) \to \ell^2 (\Gamma)$ by
\[
 \Pro_{\Gamma} \psi := \sum_{k \in \Gamma} \psi (k) \delta_k ,
\]
where the Dirac function has to be understood as an element of $\ell^2 (\Gamma)$.
Note that the corresponding embedding $\Inc_{\Gamma} := (\Pro_{\Gamma})^* : \ell^2 (\Gamma) \to \ell^2 (\ZZ^d)$ is given by
\[
\Inc_{\Gamma} \phi := \sum_{k \in \Gamma} \phi (k) \delta_k ,
\]
where here the Dirac function has to be understood as an element of $\ell^2 (\ZZ^d)$.
For an arbitrary set $\Gamma \subset \ZZ^d$ we define the restricted operators $\Delta_\Gamma, V_{\omega , \Gamma}, H_{\omega , \Gamma}:\ell^2 (\Gamma) \to \ell^2 (\Gamma)$ by $\Delta_\Gamma := \Pro_\Gamma \Delta \Inc_\Gamma$, $V_{\omega , \Gamma} := \Pro_\Gamma V_\omega \Inc_\Gamma$ and
\[
 H_{\omega , \Gamma} := \Pro_\Gamma H_\omega \Inc_\Gamma = -\Delta_\Gamma + \lambda V_{\omega , \Gamma} .
\]
Furthermore, we define $G_\Gamma (z) := (H_\Gamma - z)^{-1}$ and $G_\Gamma (z;x,y) := \bigl\langle \delta_x, G_\Gamma (z) \delta_y \bigr\rangle$ for $z \in \CC \setminus \sigma (H_\Gamma)$ and $x,y \in \Gamma$.
\subsection{Discrete alloy-type model}
Of particular interest will be the case where the random variables $\eta_k$ are given by a linear combination of i.i.d.\ random variables, 
giving rise to a discrete alloy-type potential. While some abstract definitions in Section~\ref{s:regularity_pot} hold for arbitrary random fields $\eta$, our main results concern the discrete alloy-type potential.
\begin{assumption}\label{ass:alloy}
The probability space $(\Omega , \mathcal{A} , \PP)$ is given by the product space $\Omega = \times_{k \in \ZZ^d} \RR$, $\mathcal{A} = \otimes_{k \in \ZZ^d} \mathcal{B} (\RR)$ and $\PP = \otimes_{k \in \ZZ^d} \mu$, where $\mu$ is some probability measure on $\RR$. The random variables $\eta_k : (\Omega, \mathcal{A}) \to (\RR , \mathcal{B} (\RR))$, $k \in \ZZ^d$, are given by
\begin{equation} \label{eq:alloy}
 \eta_k (\omega) = \sum_{i \in \ZZ^d} \omega_i u (x-i) 
\end{equation}
for some summable function $u : \ZZ^d \to \RR$.
\end{assumption}
If Assumption~\ref{ass:alloy} is satisfied, we call the collection of random variables $\eta_k$, $k \in \ZZ^d$, given by Eq.~\eqref{eq:alloy} a \emph{discrete alloy-type potential}, the corresponding family of operators
\[
 H_\omega = -\Delta + \lambda V_\omega, \quad (V_\omega \psi)(x) = \eta_x(\omega)\psi(x), \quad \omega \in \Omega, 
\]
on $\ell^2 (\ZZ^d)$ a \emph{discrete alloy-type model}, and the function $u$ a \emph{single-site potential}. Moreover, we set $\Theta = \supp u$.

In the case where the single-site potential $u=\delta_0$, the random Hamiltonian \eqref{eq:hamiltonian} is exactly the standard Anderson model.
\section{Conditional distributions and modulus of continuity} \label{s:regularity_pot}
\subsection{Definition and main result}
Let $m \in \ZZ^d$, $Z_m^\perp = \times_{k \in \ZZ^d \setminus \{m\}} \RR$ and $\mathcal{Z}_m^\perp = \otimes_{k \in \ZZ^d \setminus \{m\}} \mathcal{B} (\RR)$. We introduce the random variable
\[
 \eta_m^\perp : (\Omega, \cA) \to (Z_m^\perp , \mathcal{Z}_m^\perp), \quad \eta_m^\perp (\omega) = (\eta_k (\omega))_{k \in \ZZ^d \setminus \{ m \}} .
\]
We denote by $\PP_{\eta_m^\perp} : Z_m^\perp \to [0,1]$ the distribution of $\eta_m^\perp$ with respect to $\PP$, i.e.\ $\PP_{\eta_m^\perp} (B) := \PP (\{\omega \in \Omega \colon \eta_m^\perp (\omega) \in B\})$.
For $m \in \ZZ^d$, $a \in \RR$ and $\epsilon > 0$ we define the conditional expectation
\[
  Y_m^{\epsilon,a} := \PP \bigl( \eta_m \in [a,a+\epsilon] \mid \eta_m^\perp \bigr) :=  \EE \bigl( \Eins_{\{ \eta_m \in [a,a+\epsilon] \}} \mid \eta_m^\perp \bigr) .
\]
A conditional expectation $Y_m^{\epsilon,a} = \EE (\Eins_{\{\eta_m \in [a,a+\epsilon]\}} \mid \eta_m^\perp )$ is 
a random variable $Y_m^{\epsilon,a} :\Omega \to [0,1]$ with the property that
\begin{enumerate}[(i)]
\item $Y_m^{\epsilon,a}$ is $\cF$-measurable, where $\cF = \sigma (\eta_m^\perp)$, and that
\item for all $A \in \cF$ we have $\EE(\Eins_{\{\eta_m \in [a,a+\epsilon]\}} \Eins_A) = \EE (Y_m^{\epsilon,a} \Eins_A)$ .
\end{enumerate}
Since 
$\Eins_{\{\eta_m \in [a,a+\epsilon] \}} \in \mathcal{L}^1 (\Omega , \mathcal{A} , \PP)$, $Y_m^{\epsilon , a}$ exists. There may exist several functions $Y_m^{\epsilon,a}$ which satisfy conditions (i) and (ii). They are called \emph{versions} of 
$\EE (\Eins_{\{\eta_m \in [a,a+\epsilon]\}} \mid \eta_m^\perp )$. Two such versions coincide $ \PP$-almost everywhere. For convenience, for each $a \in \RR$ and $\epsilon >0$ we fix one version $Y_m^{\epsilon,a}$ of the conditional expectation.
Since $Y_m^{\epsilon,a}$ is $\cF$-measurable, the factorization lemma tells us that (for each $a$ and $\epsilon$) there is a measurable function $g_m^{\epsilon , a} : (Z_m^\perp , \mathcal{Z}_m^\perp) \to (\RR , \cB (\RR) )$ 
such that $Y_m^{\epsilon,a} = g_m^{\epsilon , a} \circ \eta_m^\perp$, i.e.\ for all $\omega \in \Omega$ we have
\begin{equation} \label{eq:factor}
 Y_m^{\epsilon , a} (\omega) = g_m^{\epsilon , a}(\eta_m^\perp (\omega)) .
\end{equation}
\par
We introduce several quantities used in the literature to describe the regularity of (the conditional distribution)
of the random field $\eta_k$, $k \in \ZZ^d$. 
For $m \in \ZZ^d$ we denote by $S_m \colon [0,\infty) \to [0,1]$,
\[
S_m (\epsilon) := \sup_{a \in \RR} \PP(\{ \omega \in \Omega \colon \eta_m \in [a,a+\epsilon]\}) ,
\]
the \emph{global modulus of continuity} or \emph{the concentration function} of the distribution of $\eta_m$.
For $\Lambda \subset \ZZ^d$ and $\epsilon > 0$ we define
\[
 \hat S_\Lambda (\epsilon) := \sup_{m \in \Lambda} \sup_{a \in \RR} \esssup_{\eta_m^\perp \in Z_m^\perp}   g_m^{\epsilon,a}(\eta_m^\perp) .
\]
Here, the essential supremum refers to the measure $\PP_{\eta_m^\perp}$, that is, 
\[
 \esssup_{\eta_m^\perp \in Z_m^\perp}   g_m^{\epsilon,a}(\eta_m^\perp) = \inf \Bigl\{b \in \RR \colon \PP_{\eta_m^\perp} \bigl(\{\eta_m^\perp \in Z_m^\perp \colon g_m^{\epsilon , a} (\eta_m^\perp) > b\}\bigr) = 0  \Bigr\} .
\]
Denote by $\tilde S_m^\epsilon$ the \emph{conditional global modulus of continuity} or the \emph{conditional concentration function} of the distribution of $\eta_m$, i.e.\ 
\[
\tilde S_m^\epsilon : \Omega \to [0,1], \quad \tilde S_m^\epsilon = \sup_{a \in \RR}  Y_m^{\epsilon , a}. 
\]
Since we are taking here a supremum over an uncountable set, 
it is not clear whether the resulting function is still measurable. In fact, this depends on 
how we chose the version of the conditional expectation (for each of the uncountable many 
$a\in \RR$). We show in Lemma~\ref{lemma:measurable} that if we choose a regular version of $Y_m^{\epsilon , a}$ (which always exists since $\eta_m$ is real-valued), then 
$\tilde S_m^\epsilon $ is $\cF$-measurable. In what follows we always assume that $\tilde S_m^\epsilon $ is $\cF$-measurable and we denote by $g_m^\epsilon : (Z_m^\perp , \mathcal{Z}_m^\perp) \to (\RR , \cB (\RR))$ the measurable function
which comes up with the factorization lemma and satisfies $\tilde S_m^\epsilon = g_m^\epsilon \circ \eta_m^\perp$. Finally, we define
\[
 \tilde S_\Lambda (\epsilon) := \sup_{m \in \Lambda} \esssup_{\eta_m^\perp \in Z_m^\perp} g_m^\epsilon (\eta_m^\perp) ,
\]
where the essential supremum again refers to the measure $\PP_{\eta_m^\perp}$.
\begin{lemma} \label{lemma:measurable}
Let $(\Omega, \cA, \PP )$ be a probability space, $\mathcal{C}\subset \cA$ a $\sigma$-algebra and $X\colon \Omega \to \RR$ a random variable. 
Let further $Q\colon \Omega \times \cB(\RR) \to [0,1]$ be a regular version of the conditional 
distribution of $X$ with respect to $\mathcal{C}$.
Then for all $\epsilon > 0$ the function
\[
 \sup_{a \in \RR} Q (\cdot  , [a,a+\epsilon]) : \Omega \to [0,1]
\]
is $\mathcal{C}$-measurable.
\end{lemma}
For the proof we will use results on the regular version of the condition a distribution of a random variable with respect to a sub-$\sigma$-algebra.
These can be found, e.g., in \S 44 of \cite{Bauer-91}.
\begin{proof}[Proof of Lemma~\ref{lemma:measurable}]
For each $\varepsilon >0$ and $a \in \RR$ 
\[
 \Omega \ni \omega \mapsto Q(\omega, [a, a +\varepsilon])
\]
is $\mathcal{C}$-measurable. Consequently, for each   $\varepsilon >0$ 
\[
 \sup_{b,\delta \in \QQ, \delta \in [0,\varepsilon]} Q(\omega, [b, b +\delta])
\]
is $\mathcal{C}$-measurable as well. It remains to show
 \[
 \sup_{a\in \RR} Q(\omega, [a, a +\varepsilon]) =
 \sup_{b,\delta \in \QQ, \delta \in [0,\varepsilon]} Q(\omega, [b, b +\delta]) .
\]
 Fix $c\in \RR$. 
Since $Q$ is a regular version of the conditional distribution we have for all $\omega \in \Omega$
\[
 Q(\omega, [c, c +\varepsilon])=
\sup_{b,\delta \in \QQ, b \geq c,\delta \geq 0, b+\delta \leq c+\varepsilon} Q(\omega, [b, b +\delta]) .
\]
(For an arbitrary version of the conditional distribution we would have this statement only for almost all $\omega$,
with the exceptional set depending on $c$.)
The last quantity equals
\[
\sup_{b,\delta \in \QQ, b \geq c,\delta \geq 0, b+\delta \leq c+\varepsilon, \delta \leq \varepsilon} Q(\omega, [b, b +\delta])
\]
and is bounded from above by 
\begin{align*}
\sup_{b,\delta \in \QQ, b \geq c,\delta \geq 0, \delta \leq \varepsilon} Q(\omega, [b, b +\delta])
\leq 
\sup_{b,\delta \in \QQ, \delta \geq 0, \delta \leq \varepsilon} Q(\omega, [b, b +\delta])
\\
\leq 
\sup_{b\in \QQ} Q(\omega, [b, b +\varepsilon])
\leq 
\sup_{b\in \RR} Q(\omega, [b, b +\varepsilon]) . 
\end{align*}
This completes the proof.
\end{proof}
The papers \cite{DreifusK-91,AizenmanM-93,AizenmanG-98,Hundertmark-00,AizenmanSFH-01,Hundertmark-08} make use of certain regularity conditions, formulated in terms of the conditional modulus of continuity, which are used to derive localization for Anderson-type models as in \eqref{eq:hamiltonian} with correlated potentials. Since the regularity conditions of the mentioned papers are all in the same flavor, we formulate exemplary the condition from \cite{AizenmanSFH-01}.
\begin{condition}\label{ass:conditional}
The collection $\eta_k$, $k\in \ZZ^d$,
is said to be (uniformly) \emph{$\tau$-H\"older continuous} for $\tau \in(0,1]$ if there is a constant $C$ such that for all $\epsilon > 0$
\[
 \hat S_{\ZZ^d} (\epsilon) := \sup_{m \in \ZZ^d} \sup_{a \in \RR} \esssup_{\eta_m^\perp \in Z_m^\perp} g_m^{\epsilon , a} (\eta_m^\perp) \leq C \epsilon^\tau .
\]
\end{condition}
Our main results on the modulus of continuity is the following theorem. 
It applies to a class of discrete alloy-type potentials, including a case where the measure $\mu$ has unbounded support.
\begin{theorem} \label{thm:result_modulus} Let Assumption~\ref{ass:alloy} be satisfied, $d = 1$ and either
\begin{enumerate}[(a)]
 \item $\Theta = \{0,\ldots , n-1\}$ for some $n \in \NN$, $\sup \supp \mu = 1$ and $\inf \supp \mu = 0$, \textbf{or}
\item $\Theta = \{-1,0\}$, $u(0) = 1$, $\lvert u(1) \rvert^2 = 1$ and $\mu$ be the normal distribution with mean zero and variance $\sigma^2$.
\end{enumerate}
Then for any $\Lambda \subset \ZZ$ and any $\epsilon > 0$ we have
 \begin{equation}\label{eq:first_statement}
   \hat S_\Lambda (\epsilon) := \sup_{m \in \Lambda} \sup_{a \in \RR} \esssup_{\eta_m^\perp \in Z_m^\perp} g_m^{\epsilon , a} (\eta_m^\perp) = 1 
\end{equation}
and
\begin{equation}\label{eq:second_statement}
 \tilde S_\Lambda (\epsilon) := \sup_{m \in \Lambda} \esssup_{\eta_m^\perp \in Z_m^\perp} g_m^{\epsilon} (\eta_m^\perp) = 1 .
\end{equation}

\end{theorem}
The above Theorem shows that Condition~\ref{ass:conditional} is not satisfied for the discrete alloy-type potential, if any of the two cases (a) or (b) holds.
This is in sharp contrast to the fact that the concentration function $S_m$ of the distribution of such $\eta_m$ 
may be very well $\tau$-H\"older continuous, as the following example shows.
\begin{example}
 Let Assumption~\ref{ass:alloy} be satisfied, $d=1$, $\Theta = \{0,1\}$, $u(0)=u(1)=1$ and $\mu$ be the uniform distribution on $[0,1]$, which is a special case of case (a) in Theorem~\ref{thm:result_modulus}. Then we have for $m \in \ZZ$ and $\epsilon > 0$
 \begin{align*}
 S_m (\epsilon) = \sup_{a \in \RR} \PP (\{\omega_m + \omega_{m-1} \in [a,a+\epsilon]\}) 
= \begin{cases}                                                                                                                                                                                 
  \epsilon - \frac{\epsilon^2}{4} & \text{if $\epsilon \in (0,2]$,} \\
  1                               & \text{if $\epsilon > 2$.}
  \end{cases}
 \end{align*}
\end{example}
If one considers finite a volume restriction $H_{\omega , \Lambda_L}$, $\Lambda_L = \{y \in \ZZ^d \colon \lvert y \rvert_\infty \leq L\}$, an analogue to Condition~\ref{ass:conditional} which is sufficient for localization would be the following: 
There is some $\tau \in(0,1]$ and a constant $C$ such that
\begin{equation} \label{eq:conditional} 
\sup_{L \in \NN} \sup_{m \in \Lambda_L} \sup_{a \in \RR} \esssup_{(\eta_k)_{k \in \Lambda_L \setminus \{m\}}} \PP \left(\eta_m \in [a,a+\epsilon] \mid (\eta_k)_{k \in \Lambda_L \setminus \{m\}} \right) \leq C \epsilon^\tau .
\end{equation}
As can be seen from the proof of Theorem~\ref{thm:result_modulus}, this condition is also not satisfied for the discrete alloy-type potential if any of the two cases (a) or (b) holds. On the contrary, the proof of Theorem~\ref{thm:result_modulus} suggests that Condition~\ref{ass:conditional} is satisfied for the discrete alloy-type potential, if $\Theta = \{-1,0\}$, $u(0) = 1$, $\lvert u(1) \rvert^2 \not = 1$ and $\mu$ is the normal distribution, see Proposition~\ref{prop:gaussian1} and Remark~\ref{remark:gaussian1} below. 
\par
The result of Theorem~\ref{thm:result_modulus} also shows that the key Lemma~3 in \cite{Klopp-12} is not correct.
Lemma~3 in \cite{Klopp-12} states (in our notation) that the conditional distributions of the random variables $\eta_m$ exhibits qualitatively the same regularity as the distributions of the random variables $\omega_m$. 
\par
The proof of Theorem~\ref{thm:result_modulus} is split into two parts. First we consider  in Subsection~\ref{ss:ecp}
elementary conditional probabilities (conditioned on an event, not on a $\sigma$-algebra)
and derive appropriate bounds. Thereafter we show how to transfer these bounds to 
probabilities conditioned on a  $\sigma$-algebra in Subsection \ref{ss:proof}.
\subsection{Elementary conditional probabilities} 
\label{ss:ecp}

\begin{proposition} \label{lemma:negexample}
 Let Assumption~\ref{ass:alloy} be satisfied, $d = 1$, $\Theta = \{0,1,\dots,n-1\}$ for some $n \in \NN$, $\inf \supp \mu = 0$ and $\sup \supp \mu = 1$. There are constants $c,m,s^+ \in (-\infty,\infty)$, depending only on $u$, such that for all $\delta > 0$ and $\delta \geq \delta' > 0$
\[
 \PP \bigl ( \eta_0 \in [m-c\delta , m + c \delta] \mid \eta_{-1},\eta_{n-1} \in [s^+ - \delta' , s^+]  \bigr ) = 1 .
\]
The values of the constants $c$, $m$ and $s^+$ can be inferred from the proof.
\end{proposition}
Notice that, under the assumptions of Proposition~\ref{lemma:negexample}, $\eta_{-1}$ and $\eta_{n-1}$ are stochastically independent and $\PP ( \eta_{-1},\eta_{n-1} \in [s^+ - \delta' , s^+] )  > 0$, where $s^+$ is defined in the proof of Proposition~\ref{lemma:negexample}.
\begin{proof}[Proof of Proposition~\ref{lemma:negexample}]
Let $\Theta^+ := \{k \in \ZZ : u (k) > 0\}$, $\Theta^- := \{k \in \ZZ^1 : u (k) < 0\}$, $u_{\rm max} = \max_{k \in \Theta} \lvert u(k) \rvert$, $u_{\rm min} = \min_{k \in \Theta} \lvert u(k)\rvert$ and $s^+ = \sum_{k \in \Theta^+} u(k)$. Let us further introduce two subsets of $\Theta$ which are important in our study. The first one is
\[
\Theta_1 =
\begin{cases}
\Theta^+ + 1 & \text{if $n-1 \not \in \Theta^+$}, \\[1ex]
\left ( (\Theta^+ + 1)\cap \Theta \right ) \cup \{0\} & \text{if $n-1 \in \Theta^+$},
\end{cases}
\]
with $\Theta^+ + 1 = \{k\in \NN : (k-1)\in \Theta^+ \}$. The second subset is the complement $\Theta_0 = \Theta \setminus \Theta_1$. To end the proof we show the following interval arithmetic result:
\par
\textit{Let $\delta \geq \delta' > 0$ and $\eta_{-1},\eta_{n-1} \in [s^+ - \delta' , s^+]$. Then
\begin{equation} \label{eq:proof0}
\eta_{0} \in [m-c\delta' , m + c \delta'] \subset [m-c\delta , m + c \delta]
\end{equation}
with $c = n u_{\rm max} / u_{\rm min}$ and $m = \sum_{k \in \Theta_1} u (k)$.}
\par
We divide the proof of \eqref{eq:proof0} into three parts. The first step is to argue that
\begin{equation} \label{eq:proof1}
\omega_{-1-k} \in 
\begin{cases}
\bigl[ 1 - \frac{\delta'}{u_{\rm min}} , 1 \bigr]  & \text{for $k \in \Theta^+$}, \\[1ex]
\bigl[ 0, \frac{\delta'}{u_{\rm min}} \bigr]   & \text{for $k \in \Theta^-$}.
\end{cases} 
\end{equation}
For the proof of the first part of \eqref{eq:proof1} we use the assumption $\eta_{-1} \geq s^+ - \delta'$ and obtain
\[
s^+ - \delta' \leq \eta_{-1} = \sum_{k \in \Theta} u(k) \omega_{-1-k} \leq \sum_{k \in \Theta^+} u(k) \omega_{-1-k} , 
\]
and hence $\sum_{k \in \Theta^+} u(k) (1- \omega_{-1-k}) \leq \delta'$. We conclude that for all $k \in \Theta^+$ we have $u(k) (1-\omega_{-1-k}) \leq \delta'$ which gives the first part of \eqref{eq:proof1}. For the proof of the second part of \eqref{eq:proof1} we use again the assumption $\eta_{-1} \geq s^+ - \delta'$ and obtain
\begin{equation*}
\sum_{k \in \Theta^+} u(k) \omega_{-1-k} - \delta'  \leq s^+ - \delta'  \leq  \eta_{-1} = \sum_{k \in \Theta^+} u(k) \omega_{-1-k} + \sum_{k \in \Theta^-} u(k) \omega_{-1-k}
\end{equation*}
which gives $-\delta' \leq \sum_{k \in \Theta^-} u(k) \omega_{-1-k}$. Thus, for all $k \in \Theta^-$ we have $\omega_{-1-k} \leq -\delta' / u(k) = \delta' / \lvert u(k) \rvert$ which gives the second part of \eqref{eq:proof1}. In a second step we argue that 
\begin{equation} \label{eq:proof2}
\omega_{-k+n-1} \in 
\begin{cases}
\bigl[ 1 - \frac{\delta'}{u_{\rm min}} , 1 \bigr]  & \text{for $k \in \Theta^+$}, \\[1ex]
\bigl[ 0, \frac{\delta'}{u_{\rm min}} \bigr]   & \text{for $k \in \Theta^-$} . 
\end{cases} 
\end{equation}
The proof of \eqref{eq:proof2} can be done in analogy to the proof of \eqref{eq:proof1}, but using the assumption $\eta_{n-1} \geq s^+ - \delta'$. 
In a third step we ask the question for which $k \in \Theta$ we have $\omega_{-k} \in [1-\delta' / u_{\rm min} , 1]$. Using the definition of the set $\Theta_1$ we find with \eqref{eq:proof1} and \eqref{eq:proof2} that
\begin{equation} \label{eq:proof3}
\omega_{-k} \in 
\begin{cases}
\bigl[ 1 - \frac{\delta'}{u_{\rm min}} , 1 \bigr]  & \text{for $k \in \Theta_1$}, \\[1ex]
\bigl[ 0, \frac{\delta'}{u_{\rm min}} \bigr]   & \text{for $k \in \Theta_0$} . 
\end{cases} 
\end{equation}
Now, the desired result \eqref{eq:proof0} follows from \eqref{eq:proof3} and the decomposition
\[
\eta_{0} = \sum_{k \in \Theta} u(k) \omega_{-k} = \sum_{k \in \Theta_1} u(k) \omega_{-k} + \sum_{k \in \Theta_0} u(k) \omega_{-k} .
\]
Hence, the proof is complete.
\end{proof}
\begin{remark}
 The assumption $\inf \supp \mu = 1$ and $\sup \supp \mu = 1$ in Proposition~\ref{lemma:negexample} is not crucial. What matters is that $\supp \mu$ is a bounded set.
\end{remark}
In the case where $\supp \mu$ is an unbounded set the situation is somehow different. We illustrate the effects in the case where $\mu$ is Gaussian.
For $l \in \NN$ let $A_{l} \in \RR^{l \times l+1}$ be the matrix with coefficients in the canonical basis given by 
$A_{l} (i,i) = 1$, $A_l (i,i+1) = u(-1)$ for $i \in \{1,\dots,l\}$, 
and zero otherwise, namely
\[
A_l= \begin{pmatrix}
1 & u(-1) & & &\\
& \ddots & \ddots & &\\
 & & \ddots & u(-1) &\\
  & & & 1 & u(-1)
\end{pmatrix} \in \RR^{l \times l+1} .
\]
\begin{proposition} \label{prop:gaussian1}
Let Assumption~\ref{ass:alloy} be satisfied, $d = 1$, $l,m \geq 1$, $\Theta = \{-1,0\}$, $u(0)=1$ and $\mu$ be the normal distribution with mean zero and variance $\sigma^2$. Let further $v^+ \in \RR^{l}$ and $v^- \in \RR^{m}$. 
Then the distribution of $\eta_{0}$ conditioned on $(\eta_{k})_{k=1}^{l} = v^+$ and $(\eta_{-m+k-1})_{k=1}^{m} = v^-$ 
is Gaussian with variance
\begin{equation*} 
 \gamma = \sigma^2 \Bigl( u(-1)^2 - 1 + \frac{1}{s_m} + \frac{1}{s_{l}} \Bigr), \quad \text{where} \quad  s_l := \sum_{i=1}^{l} \bigl(u(-1)\bigr)^{2i},
\end{equation*}
and mean
\[
 m = u(-1) \left( \sum_{i=1}^{m} (A_{m} A_{m}^{\rm T})^{-1} (m,i) \, v^-_i + \sum_{i=1}^{l} (A_{l} A_{l}^{\rm T})^{-1} (1,i) \, v^+_i \right) .
\]
 \end{proposition}
\begin{remark} \label{remark:gaussian1}
 Let $l,m \geq 1$. If $\lvert u(-1) \rvert \not = 1$, Proposition~\ref{prop:gaussian1}  gives that the distribution of $\eta_0$ conditioned on fixed potential values $\eta_k$, $k \in \{-m,\dots,l\} \setminus \{0\}$, is again Gaussian with variance bounded from below by $\sigma^2 \lvert u^2 (-1) - 1 \rvert$. 
This shows that the random field $\eta_k$, $k \in \ZZ^d$, satisfies the regularity condition formulated in Ineq.~\eqref{eq:conditional} 
if $\Theta = \{-1,0\}$, $u(0)=1$, $\lvert u(-1) \rvert \not = 1$ and $\mu$ is Gaussian. 
Moreover, the regularity condition from Ineq.~\eqref{eq:conditional} is not satisfied if $\Theta = \{-1,0\}$, $u(0)=1$, $\lvert u(-1) \rvert = 1$ and $\mu$ is Gaussian.
\end{remark}
The proof of Proposition~\ref{prop:gaussian1} is based on following classical result which may be found in \cite{Port1994}.
\begin{proposition} \label{prop:cond0}
Let $X$ be normally distributed on $\RR^d$, $Y = a \cdot X$ where $a \in \RR^{d}$, and $W = B X$ where $B \in \RR^{m \times d}$. Assume $W$ has a non-singular distribution. Then the distribution of $Y$ conditioned on $W = v \in \RR^m$ is the Gaussian distribution having mean
\[
 \mathbf{E} ( Y ) + \cov (Y,W) \cov(W,W)^{-1} [v - \mathbf{E} (W)]
\]
and variance
\[
 \cov (Y,Y) - \cov (Y,W) \cov(W,W)^{-1} \operatorname{cov} (W,Y) .
\]
\end{proposition}
Notice, if we apply $A_l$ on the vector $\omega_{[x,x+l]} = (\omega_{x+k-1})_{k=1}^{l+1}$, we obtain a vector containing the potential values $\eta_{k}$, $k \in \{x,x+1,\ldots,x+l\}$. Moreover, the vector $(\eta_{x+k-1})_{k=1}^l = A_l \omega_{[x,x+l]}$ is normally distributed with mean zero and covariance $\sigma^2 A_l A_l^{\rm T}$.
The matrix $A_l A_l^T$ has the form
\[
A_l A_l^T = \begin{pmatrix}
1 + u^2(-1) 	& u(-1) 	& 		& 		\\
u(-1) 	& 1+u^2(-1) 	& \ddots 	& 		\\
 		& \ddots 	& \ddots 	& u(-1)	\\
  		& 		& u(-1) 	& 1 + u^2(-1)
\end{pmatrix} \in \RR^{l \times l} .
\]
By induction we find that the determinant of $A_l A_l^{\rm T}$ is given by
\[
 \det (A_l A_l^{\rm T}) = s_l > 0 \quad \text{where} \quad s_l = \sum_{i=1}^{l} \bigl(u(-1)\bigr)^{2i} .
\]
Since the minor $M_{11}$ and $M_{ll}$ of $A_l A_l^{\rm T}$ equals $A_{l-1} A_{l-1}^{\rm T}$ we obtain by Cramers rule for the elements $(1,1)$ and $(l,l)$ of the inverse of $A_{l-1} A_{l-1}^{\rm T}$
\begin{equation} \label{eq:inverseelement}
 (A_{l} A_{l}^{\rm T})^{-1} (1,1) = (A_{l} A_{l}^{\rm T})^{-1} (l,l) = \frac{s_{l-1}}{s_l} .
\end{equation}
\begin{proof}[Proof of Proposition~\ref{prop:gaussian1}]
 Let $X := (\omega_{-m-1+k})_{k=1}^{l+m+2} \in \RR^{m+n+2}$, $a = (a_{i})_{i=1}^{l+m+2} \in \RR^{l+m+2}$ the vector with coefficients $a_{m+1}=1$, $a_{m+2} = u(-1)$ and zero otherwise. Let us further define the block-matrix
\[
 B = \begin{pmatrix}
      A_{m} & 0 \\
      0   & A_{l}
     \end{pmatrix} \in \RR^{(m+l) \times (m+l+2) }.
\]
Notice that $Y := a \cdot X = \eta_{0}$,
\[
 A_{m} \omega_{[-m,0]} = (\eta_{-m+k-1})_{k=1}^{m}, \quad \text{and} \quad 
 A_{l} \omega_{[1,l+1]}= (\eta_{k})_{k=1}^{l} , 
\]
where $\omega_{[-m,0]} = (\omega_{-m+k-1})_{k=1}^{m+1}$ and $\omega_{[1,l+1]} = (\omega_{k})_{k=1}^{l+1}$.
Hence $W := B X$ is the $m+l$-dimensional vector containing the potentials $\eta_{k}$, $k \in \{-m, \dots, l\} \setminus \{0\}$. Notice that $Y$ and $W$ have mean zero, since $X$ has mean zero.
We apply Proposition~\ref{prop:cond0} with these choices of $X$, $Y$ and $W$, and obtain that the distribution of $\eta_{0}$ conditioned on $(\eta_{-m+k-1})_{k=1}^{m} = v^-$ and $(\eta_{k})_{k=1}^{l} = v^+$ is Gaussian with mean 
\begin{equation*}
m = \cov (Y,W) \cov(W,W)^{-1} v 
\end{equation*}
and variance 
\[
\gamma = \cov (Y,Y) - \cov (Y,W) \cov(w,w)^{-1} \operatorname{cov} (W,Y), 
\]
where $v = (v^- , v^+)^{\rm T}$. It is straightforward to calculate $\cov (Y,Y) = \sigma^2 (1+ u(-1)^2)$ and $\cov (W,Y) =  z = (z^- , z^+)^{\rm T}$, where $z^- = (0,\dots,0,\sigma^2 u(-1))^{\rm T} \in \RR^{m}$ and $z^+ = (\sigma^2 u(-1) ,0,\dots,0)^{\rm T} \in \RR^{l}$. We also have
\[
 \cov (W,W) =   \sigma^2 \begin{pmatrix}
                         A_{m} A_{m}^{\rm T} & 0 \\
                        0  & A_{l} A_{l}^{\rm T}
                       \end{pmatrix} .
\]
Hence by Eq.~\eqref{eq:inverseelement}
\begin{align*}
\gamma &= \sigma^2 (1 + u(-1)^2) - \sigma^{-2}  z^{\rm T} \begin{pmatrix}
                        ( A_{m} A_{m}^{\rm T})^{-1} & 0 \\
                        0  & (A_{l} A_{l}^{\rm T})^{-1}
                       \end{pmatrix}^{-1} z  \nonumber \\
&= \sigma^2 (1 + u(-1)^2) - \sigma^{-2} \biggl[ \sigma^4 u^2(-1) \frac{s_{m-1}}{s_m} +  \sigma^4 u^2 (-1) \frac{s_{l-1}}{s_l} \biggr] \\
&= \sigma^2 (1 + u(-1)^2) - \sigma^{2} \biggl(  1-\frac{1}{s_m} \biggr)- \sigma^{2} \biggl(  1-\frac{1}{s_m} \biggr) ,
\end{align*}
and
\begin{equation*}
 m =  \bigl[{z^-}^{\rm T}(\sigma^2 A_{m} A_{m}^{\rm T})^{-1} {v^-} + {z^+}^{\rm T}(\sigma^2 A_{l} A_{l}^{\rm T})^{-1} {v^+}\bigr] .
\end{equation*}
This proves the statement of the proposition.
\end{proof}
The case of Proposition~\ref{prop:gaussian1} where either $m$ or $l$ equals zero can be proven analogously and is indeed contained in the statement of Proposition~\ref{prop:gaussian1} in the sense that $s_0 = 1$. However, to avoid confusion let us reformulate the case $m = 0$.
\begin{proposition} \label{prop:gaussian2}
 Let Assumption~\ref{ass:alloy} be satisfied, $d = 1$, $l \geq 1$, $\Theta = \{-1,0\}$, $u(0) = 1$ and $\rho$ be the Gaussian density with mean zero and variance $\sigma^2$. Let further $v \in \RR^{l}$. Then the distribution of $\eta_{0}$ conditioned on $(\eta_{k})_{k=1}^{l} = v$ is Gaussian with variance
\begin{equation*}
 \gamma = \sigma^2 \Bigl( u(-1)^2 + \frac{1}{s_l} \Bigr) \quad \text{and mean} \quad   m = u(-1) \sum_{i=1}^{l} (A_{l} A_{l}^{\rm T})^{-1} (1,i) \, v_i .
\end{equation*}
\end{proposition}
\subsection{Proof of Theorem~\ref{thm:result_modulus}} \label{ss:proof}
\begin{proposition} \label{prop:cond_exp}
Let Assumption~\ref{ass:alloy} be satisfied, $d = 1$ and either
\begin{enumerate}[(a)]
 \item $\Theta = \{0,\ldots , n-1\}$ for some $n \in \NN$, $\sup \supp \mu = 1$, $\inf \supp \mu = 0$, $m$ be as in Lemma~\ref{lemma:negexample}, $\epsilon > 0$ and $a = m - \epsilon / 2$,
\textbf{or}
\item $\Theta = \{-1,0\}$, $u(0) = 1$, $\lvert u(1) \rvert^2 = 1$, $\mu$ be the normal distribution with mean zero and variance $\sigma^2$, $\epsilon > 0$ and $a = -\epsilon/2$.
\end{enumerate}
Then,
\[
 \esssup_{\eta_0^\perp \in Z_0^\perp} g_0^{\epsilon,a} (\eta_0^\perp) = 1 .
\]
\end{proposition}
\begin{proof}
 Assume the converse, i.e.\ $b := \esssup_{\eta_0^\perp} g_0^{\epsilon,a} (\eta_0^\perp) < 1$. 
By definition of the conditional expectation we have for all $B \in \sigma (\eta_0^\perp)$ that
\begin{equation} \label{eq:definition}
 \EE \bigl( \Eins_B \Eins_{\{\eta_0 \in [a,a+\epsilon]\}} \bigr) = \EE \bigl( \Eins_B Y_0^{\epsilon,a} \bigr) .
\end{equation}
Let $l \in \NN$, $s^+$ and $c$ be as in Lemma~\ref{lemma:negexample}, and choose
\[ 
B=
\begin{cases}
\bigl\{\omega \in \Omega \colon \eta_{-1},\eta_{n-1} \in [s^+ - \epsilon / (2c) , s^+]\} & \text{if (a) is satisfied}  \\
 \bigl\{\omega \in \Omega \colon \eta_k = 0, k \in \{-l , \ldots , l\} \setminus \{0\} \bigr\} & \text{if (b) is satisfied} 
\end{cases}
\]
which is $\sigma (\eta_0^\perp)$-measurable. Lemma~\ref{lemma:negexample} and Proposition~\ref{prop:gaussian1} tells us that the left hand side of Eq.~\eqref{eq:definition} equals
\[
 \PP ( B \cap \{\eta_0 \in [a,a+\epsilon]\} ) = 
 \begin{cases}
  1 \cdot \PP(B) & \text{if (a) is satisfied},  \\
 \mathcal{N}_{0,\gamma} ([a,a+\epsilon]) \cdot \PP(B) & \text{if (b) is satisfied} ,
 \end{cases}
\]
where $\gamma = \sigma^2 (2/l)$.
Here, $\mathcal{N}_{0,\gamma}$ denotes the normal distribution with mean zero and variance $\gamma$. Now we choose $l$ large enough, such that $\mathcal{N}_{0,\gamma} ([a,a+\epsilon]) > b$.
For the right hand side of Eq.~\eqref{eq:definition} we use the factorized version \eqref{eq:factor} of $Y_0^{\epsilon , a}$ and obtain by substitution
\[
 \EE \bigl( \Eins_B Y_0^{\epsilon,a} \bigr) = \int_{Z_{0}^\perp} \Eins_{B'} (\eta_0^\perp) g_0^{\epsilon , a} (\eta_0^\perp) \drm \PP_{\eta_0^\perp} (\eta_0^\perp) ,
\]
where 
\[
B' = \{\eta_0^\perp \in Z_0^\perp \colon \eta_{-1},\eta_{n-1} \in [s^+ - \epsilon / (2c) , s^+]\}  .
\]
Since $b < 1$ by our assumption we obtain 
\[
\EE ( \Eins_B Y_0^{\epsilon,a} ) \leq b \PP_{\eta_0^\perp} (B') = b \PP(B) <
\begin{cases}
  1 \cdot \PP(B) & \text{if (a) is satisfied},  \\
 \mathcal{N}_{0,\gamma} ([a,a+\epsilon]) \cdot \PP(B) & \text{if (b) is satisfied} .
\end{cases}
\]
This is a contradiction to Eq.~\eqref{eq:definition}.
\end{proof}
\begin{proof}[Proof of Theorem~\ref{thm:result_modulus}]
The first equality \eqref{eq:first_statement} follows from translation invariance and Proposition~\ref{prop:cond_exp}.
For the second statement \eqref{eq:second_statement} we use the pointwise inequality $g_0^\epsilon (\eta_0^\perp) \geq g_0^{\epsilon , a} (\eta_0^\perp)$. If we take first the essential supremum with respect to $\eta_0^\perp$ and then supremum with respect to $a$ on both sides, we obtain using Proposition~\ref{prop:cond_exp}
\[
 \esssup_{\eta_0^\perp \in Z_0^\perp} g_0^\epsilon (\eta_0^\perp) \geq 1 .
\]
The result now follows by translation invariance.
\end{proof}
\section{How regularity properties turn into regularity of spectral data}
\label{s:regularity}
Throughout Section \ref{s:regularity}, \ref{s:physical_implications} and \ref{s:ESS} we assume that Assumption~\ref{ass:alloy} is satisfied, i.e.\ the random field $\eta_k:(\Omega,\cA) \to (R, \cB (\RR))$, $k \in \ZZ^d$, is the discrete alloy-type potential given in Eq.~\eqref{eq:alloy}.
Next we list several additional regularity assumptions which may hold or not hold. All of them can be interpreted as assumptions of the distribution of the stochastic process $\eta_m, m\in\ZZ^d$.
\begin{assumption} \label{ass:ubar}
 The measure $\mu$ has compact support, a probability density $\rho \in W^{1,1} (\RR)$, $\Theta$ is finite and the single-site potential satisfies $\overline u = \sum_{k \in \ZZ^d} u (k) > 0$.
\end{assumption}
\begin{assumption} \label{ass:monotone} $\Theta$ is a finite set, the measure $\mu$ has bounded support and a probability density $\rho \in L^\infty (\RR)$, 
and the function $u$ satisfies $u (k) > 0$ for all $k \in \partial^{\rm i} \Theta:= \{k \in \Theta\mid k \ \text{has less than $2d$ neighbors in}\ \Theta\} $.
\end{assumption}
\begin{assumption} \label{ass:exponential} The measure $\mu$ has bounded support and a probability density $\rho \in \operatorname{BV} (\RR)$ and there are constants $C,\alpha>0$ such that for all $k \in \ZZ^d$ we have $\lvert u(k) \rvert \leq C \mathrm{e}^{-\alpha \lVert k \rVert_1}$.
\end{assumption}
If Assumption~\ref{ass:exponential} is satisfied, we define a constant $N$ as follows. For $\delta \in (0, 1-\mathrm{e}^{-\alpha})$ we consider the to $u$ associated generating function $F : D_\delta \subset \CC^d \to \CC$, 
\begin{equation*}
D_\delta = \{ z \in \CC^d : \lvert z_1 - 1 \rvert < \delta , \ldots , \lvert z_d - 1 \rvert < \delta \}, \quad F(z) = \sum_{k \in \ZZ^d} u(-k) z^k .
\end{equation*}
Notice that the sum $\sum_{k \in \ZZ^d} u(-k) z^k$ is normally convergent in $D_\delta$ by our choice of $\delta$ and the exponential decay condition of Assumption~\ref{ass:exponential}. By Weierstrass' theorem, $F$ is a holomorphic function. Since $F$ is holomorphic and not identically zero, we have $(D_z^I F) (\mathbf{1}) \not = 0$ for at least one $I \in \NN_0^d$. Therefore, there exists a multi-index $I_0 \in \NN_0^d$ (not necessarily unique), such that we have
\begin{equation} \label{eq:cF}
  (D_z^I F)({\mathbf 1}) = 
		\begin{cases} 
			c_u \neq 0, & \text{if $I = I_0$,} \\
			0,          & \text{if $I < I_0$.} 
		\end{cases}  
\end{equation}
Such a $I_0$ can be found by diagonal inspection: Let $n \ge 0$ be the largest integer such that $D_z^IF({\mathbf 1}) = 0$ for all $\lVert I \rVert_1 < n$.  Then choose a multi-index $I_0 \in \NN_0^d$, $\lvert I_0 \rvert_1 = n$ with $(D_z^{I_0} F)({\mathbf 1}) \neq 0$. We finally set $N = \lvert I_0 \rvert_1$.
\begin{assumption}
\label{ass:circulant}
Assume that $\Theta$ is a finite set, the Fourier transform $\hat u \colon [0,2\pi)^d \to
\mathbb{C}$ of $u$, i.e.
\[
 \hat u (\theta) = \sum_{k \in \mathbb{Z}^d} u(k) {\rm e}^{\ensuremath{{\mathrm{i}}}k \cdot \theta} ,
\]
does not vanish, and that the measure $\mu$ has bounded support and a density $\rho \in W^{2,1} (\mathbb{R})$.
\end{assumption}
If Assumption~\ref{ass:circulant} is satisfied, we define the constant $C_u$ as follows. Let $A:\ell^1 (\mathbb{Z}^d) \to \ell^1 (\mathbb{Z}^d)$ be the linear operator whose coefficients in the canonical orthonormal basis are given by $A(j,k) = u(j-k)$ for $j,k \in \mathbb{Z}^d$. Since $u$ has compact support, the operator $A$ is bounded. 
If $\hat u$ does not vanish (as required by Assumption \ref{ass:circulant}), the operator $A$ has a bounded inverse by the so-called $1/f$-Theorem of Wiener and we have
\begin{equation} \label{eq:Cu}
C_u := \lVert A^{-1} \rVert_1 < \infty ,
\end{equation}
see \cite{Veselic-10b} for details.
\par
We list several results on the regularity of spectral data under (some of) the above conditions on the stochastic process defining the random potential.
While these results by themselves are probabilistic statements, describing the regularity of push-forward (or image) measures, they are of crucial importance
for the study of spectral properties of random Schr\"odinger operators. This is described explicitly in the subsequent Section~\ref{s:physical_implications}.
\par
The first result concerns the uniform boundedness of the average of a fractional power of the Green function.
It is sometimes called a-priori bound of the fractional moment method. 
\begin{theorem}[\cite{ElgartTV-11}] \label{thm:finite_fm}
 Let $\Lambda \subset \ZZ^d$ finite, $s \in (0,1)$ and Assumption~\ref{ass:ubar} be satisfied. Then we have for all $x,y \in \Lambda$ and $z \in \CC \setminus \RR$
 \[
  \EE \left( \lvert G_{\omega , \Lambda} (z;x,y) \rvert^s \right) \leq \frac{8}{(\overline u)^s} \frac{s^{-s}}{1-s} \lVert \rho' \rVert_{L^1}^s C^s \frac{1}{\lambda^s},
 \]
 where 
 \[
  C = \left( \frac{\mathrm{e}^c + 1}{\mathrm{e}^c - 1} \right)^d \quad \text{and} \quad c = \frac{1}{\operatorname{diam} \Theta} \ln \left( 1 + \frac{\overline u}{2 \lVert u \rVert_{\ell^1}} \right) .
 \]
\end{theorem}
While the last statement is uniform in the lattice points $x,y$ it does not exploit the intuition that the Green function should decay, 
as the $|x-y|$ grows. Such a statement is given, for discrete alloy-type models, in the following theorem. It is the core of the fractional moment method. 
\begin{theorem}[\cite{ElgartTV-11}] \label{thm:fmb}
 Let $\Gamma \subset \ZZ^d$, $s \in (0,1/3)$ and suppose that Assumption~\ref{ass:monotone} is satisfied. Then for a sufficiently large $\lambda$ there are constants $C,m \in (0,\infty)$, depending only on $d$, $\rho$, $u$, $s$ and $\lambda$, such that for all $z \in \CC \setminus \RR$ and all $x,y \in \Gamma$
 \begin{equation} \label{eq:fmb}
  \EE \left( \lvert G_{\omega , \Gamma} (z;x,y) \rvert^{s / 2 \lvert \Theta \rvert} \right) \leq C \mathrm{e}^{-m \lvert x-y \rvert} .
 \end{equation}
\end{theorem}
As we will comment below, exponential bounds of the type \eqref{eq:fmb} allow to conclude spectral localization.
\par
For $L > 0$ we denote by $\Lambda_L = \{y \in \ZZ^d \colon \lvert y \rvert_\infty \leq L\}$ the cube centered at the origin.
The following result is a bound on the expected number of eigenvalues in a given energy region, for a finite cube random Hamiltonian. 
\begin{theorem}[\cite{PeyerimhoffTV-11,LeonhardtPTV-13}] \label{thm:wegner}
 Let Assumption~\ref{ass:exponential} be satisfied and $\lambda > 0$. Then there exists $C_{\rm W} > 0$ depending only on $u$, such that for any $L > 0$ and any bounded interval $I \subset \RR$
 \[
  \EE \left( \operatorname{Tr} \chi_I (H_{\omega , \Lambda_L}) \right) \leq \lambda^{-1} C_{\rm W} \lVert \rho \rVert_{\rm Var} \lvert I \rvert (2l+1)^{2d + N} ,
 \]
 where $N$ is defined by Eq.~\eqref{eq:cF}.
\end{theorem}
The above inequality shows the Lipschitz-continuity of the distribution function $E\mapsto \EE \left( \operatorname{Tr} \chi_{(-\infty,E]} (H_{\omega , \Lambda_L}) \right)$.
Hence, we have result about regularity of spectral data. 
\par
The following theorem is a so-called Minami-estimate. It generalizes the key estimate in \cite{Minami-96}.
Minami's result applies to the standard Anderson model, while our theorem concerns also certain correlated 
random potentials. 
\begin{theorem}[\cite{TautenhahnV-13}] \label{thm:minami}
Let $\Lambda \subset \mathbb{Z}^d$ be finite and Assumption \ref{ass:circulant} satisfied. 
Then we have for all $x,y \in \Lambda$ with $x \not = y$, all $z \in \mathbb{C}$ with $\Im z > 0$ and all $\lambda > 0$
\[
 \mathbb{E} \left( \det \left\{ \Im \begin{pmatrix}
  G_{\omega , \Lambda} (z;x,x) & G_{\omega , \Lambda} (z;x,y) \\
  G_{\omega , \Lambda} (z;y,x) & G_{\omega , \Lambda} (z;y,y)
 \end{pmatrix}  \right\} \right) \leq \left(\frac{\pi}{\lambda}\right)^2 C_{\rm Min} ,
\]
where 
\[
 C_{\rm Min} = \frac{C_u^2}{4} \max\{\lVert \rho' \rVert_1^2 , \lVert \rho'' \rVert_1\} 
\]
and $C_u$ is the constant from Eq.~\eqref{eq:Cu}.
\end{theorem}
Minami's estimate  has an important corollary, a bound on the probability of finding at least two eigenvalues of $H_{\omega , \Lambda}$ in a given energy interval. 
\begin{corollary} \label{cor:minami}
 Let Assumption \ref{ass:circulant} be satisfied, $\Lambda \subset \mathbb{Z}^d$ finite and $I \subset \mathbb{R}$ be a bounded interval. Then we have for all $\lambda > 0$
\begin{align}
\label{eq:Cebysev}
 \mathbb{P} \bigl\{ \operatorname{Tr} \chi_{I} (H_{\omega , \Lambda}) \geq 2 \bigr\} 
&\leq 
\frac{1}{2}\mathbb{E} \bigl( (\operatorname{Tr} \chi_{I} (H_{\omega , \Lambda}))^2  - \operatorname{Tr} \chi_{I} (H_{\omega , \Lambda}) \bigr) 
\\ \nonumber
&\leq 
\frac{1}{2}\left(\frac{\pi}{\lambda}\right)^2 C_{\rm Min} \lvert I \rvert^2 \lvert \Lambda \rvert^2 .
\end{align}
\end{corollary}
Thus we can control the probability that two eigenvalues fall close to each other: Again a regularity statement for spectral data. 
\section{Physical implications of the regularity of spectral data}
\label{s:physical_implications}
One motivation for proving the regularity results of spectral data is that they are the main ingredient for localization proofs. There are different signatures of localization. We discuss two of them: spectral localization and Poisson statistics. Spectral localization or Anderson localization is the phenomenon that there are energy intervals $I$ such that for almost all configurations of the randomness, the spectrum of $H_\omega$ consists only of eigenvalues. 
\begin{definition}
 Let $I \subset \RR$. We say that $H_\omega$ exhibits exponential localization in $I$ if, for almost all $\omega \in \Omega$, $\sigma_{\rm c} (H_\omega) \cap I = \emptyset$ and the eigenfunctions corresponding to the eigenvalues of $H_\omega$ in $I$ decay exponentially. 
If $I = \RR$ we simply say that $H_\omega$ exhibits exponential localization. 
\end{definition}
Beside this spectral interpretation of localization there are also interpretations from the dynamical point of view. Since we put our focus here on spectral localization we do not give a definition here. 
However, for the discussion of various notions of dynamical localization we refer to \cite{Klein-08}.
\par
In space dimension $d>1$ there are exactly two methods to prove localization: the multiscale analysis \cite{FroehlichS-83,FroehlichMSS-85} and the fractional moment method \cite{AizenmanM-93}. The output of the fractional moment method is the exponential decay of an averaged fractional power of the Green function, i.e.\ an inequality of the form \eqref{eq:fmb}. There is a variety of methods for concluding localization from this so-called fractional moment bound \eqref{eq:fmb}, for example using the Simon-Wolff criterion \cite{AizenmanM-93}, via the RAGE-theorem \cite{Graf-94,Hundertmark-00}, using a method called eigenfunction correlators \cite{AizenmanENSS-06}, or going the way via the output of multiscale analysis \cite{ElgartTV-10,ElgartTV-11}. In this sense, it is not surprising that Theorem~\ref{thm:fmb} yields the following localization result.
\begin{theorem}[\cite{ElgartTV-11}] \label{thm:loc1}
Let Assumption~\ref{ass:monotone} be satisfied and $\lambda$ sufficiently large. Then, for almost all $\omega \in \Omega$, $H_\omega$ exhibits exponential localization. 
\end{theorem}
The multiscale ananysis is an induction argument which shows the exponential decay of the Green function with high probability on larger and larger scales. The induction anchor is the so-called initial length scale estimate. The main ingredient for the induction step is a Wegner estimate, which is formulated for our specific model in Theorem~\ref{thm:wegner}. Since the initial length scale estimate follows from a Wegner estimate in the case of large disorder, we obtain the following improvement of Theorem~\ref{thm:loc1}.
\begin{theorem}[\cite{LeonhardtPTV-13}] \label{thm:loc2}
 Let Assumption~\ref{ass:exponential} be satisfied and $\lambda$ be sufficiently large. Then, for almost all $\omega \in \Omega$, $H_\omega$ exhibits exponential localization.
\end{theorem}
Due to the lack of monotonicity it is not possible by standard methods to obtain an initial length scale estimate in the case of small disorder $\lambda > 0$ under the general Assumption~\ref{ass:exponential}. However, if the single-site potential has only a small negative part, it is possible to deduce an initial length scale estimate at the bottom of the spectrum by using perturbative arguments. This has been implemented for compactly supported single-site potentials in the continuous setting in \cite{Veselic-02a} and adapted to exponentially decaying (not compactly supported) single-site potentials in the discrete setting in \cite{LeonhardtPTV-13}. Together with the Wegner estimate from Theorem~\ref{thm:wegner} one obtains localization via multiscale analysis in the weak disorder regime.
\begin{assumption} \label{ass:small_neg}
We say that Assumption~\ref{ass:small_neg} is satisfied for $\delta > 0$, if there exists a decomposition $u = u_+ - \delta u_-$ with $u_+ , u_- \in \ell^1 (\ZZ^d ; \RR^+_0)$, and $\lVert u_- \rVert_1 \leq 1$. 
For the measure $\mu$ we assume $\supp \mu = [0,\omega_+]$ for some $\omega_+ > 0$.
\end{assumption}
 \begin{theorem}[\cite{LeonhardtPTV-13}] \label{thm:loc:weak}
Let Assumption~\ref{ass:exponential} and \ref{ass:ubar} be satisfied and $\lambda > 0$. Then there exists $\delta > 0$ and $\epsilon > 0$, such that if Assumption~\ref{ass:small_neg} is satisfied for $\delta$, then, for almost all $\omega \in \Omega$, $H_\omega$ exhibits exponential localization in $[-\epsilon , \epsilon]$.
\end{theorem}
See also \cite{CaoE-12} for a more general result in three space dimensions.

Another signature of localization is Poisson statistics. Physicists expect that there is no level repulsion of energy levels in the localized regime. 
This manifests itself in the sense that the point process associated to the rescaled eigenvalues of $H_{\omega , \Lambda_L}$ converges to a Poisson process. 
\par
To be more precise, we introduce all the basic definitions.
Let $L \in \mathbb{N}$ and and $E_1^\omega (\Lambda_L) \leq E_2^\omega (\Lambda_L) \leq \ldots \leq E_{\lvert \Lambda_L \rvert}^\omega (\Lambda_L)$ be the eigenvalues of $H_{\omega , \Lambda_L}$ repeated according to multiplicity. Since $(H_\omega)_\omega$ is an ergodic family of random operators, the IDS exists as a (non-random) distribution function $N :  \mathbb{R} \to [0,1]$, satisfying for almost all $\omega \in \Omega$
\[
 N(E) = \lim_{L \to \infty} \frac{1}{\lvert \Lambda_L \rvert} \# \{ j \in \mathbb{N} \colon E_j^\omega (\Lambda_L) \leq E \},
\]
at all continuity points of $N$. In particular, if Assumption \ref{ass:circulant} is satisfied, the IDS is known to be Lipschitz continuous \cite{Veselic-10b}. 
Let us now introduce a second hypothesis which may be interpreted as a quantitative growth condition on the IDS or a positivity assumption on the density of states measure.
\begin{assumption} \label{ass:pos}
 Let $E_0 \in \mathbb R$ and $\kappa \geq 0$. We say that Assumption (Pos) is satisfied for $E_0$ and $\kappa $ if for all $a<b$ there exists $C,\epsilon_0 > 0$ such that for all $\epsilon \in (0,\epsilon_0)$ there holds
\[
 \lvert N(E_0 + a \epsilon) - N(E_0 + b \epsilon) \rvert \geq C \epsilon^{1 + \kappa } .
\]
\end{assumption}
For $E_0 \in \mathbb{R}$ we consider the rescaled spectrum $\xi^\omega = (\xi_j^\omega)_{j=1}^{\lvert \Lambda_L \rvert}$, defined by
\begin{equation} \label{eq:scaleev}
 \xi_j^\omega = \xi_j^\omega (L , E_0) = \lvert \Lambda_L \rvert \left(N(E_j^\omega (\Lambda_L)) - N(E_0)\right), \quad j=1,\ldots , \lvert \Lambda_L \rvert ,
\end{equation}
and the associated point process $\Xi : \Omega \to \mathcal{M}_{\rm p}$ given by
\begin{equation} \label{eq:process}
 \Xi^\omega = \Xi^\omega_{L , E_0}  = \sum_{j=1}^{\lvert \Lambda_L \rvert} \delta_{\xi_j^\omega} ,
\end{equation}
where $\delta_x$ is the Dirac measure concentrated at $x$ and $\mathcal{M}_{\rm p}$ is the set of all integer valued Radon measures on $\mathbb{R}$. 
A point process $\Upsilon$ is called Poisson point process with intensity measure $\mu$ if 
\[
 \mathbb{P} \bigl\{ \omega \in \Omega \colon \Upsilon^\omega (A) = k \bigr\} = \mathrm{e}^{-\mu(A)} \frac{\mu (A)^k}{k!}, \quad k=1,2,\ldots 
\]
holds for each bounded Borel set $A \in \mathcal{B} (\mathbb{R})$
and for disjoint sets $A_1 , \ldots , A_n \in \mathcal{B} (\RR)$, the random variables $\Upsilon (A_1), \ldots , \Upsilon (A_n)$ are independent.
Let $\Upsilon_n : \Omega \to \mathcal M_{\rm p}$, $n \in \mathbb N$, be a sequence of point processes defined on a probability space $(\Omega , \mathcal A , \mathbb P)$. This sequence is said to converge weakly to a point process $\Upsilon : \tilde \Omega \to \mathcal M_{\rm p}$ defined on a probability space $(\tilde \Omega , \tilde{\mathcal{A}} , \tilde{\mathbb{P}})$, if and only if for any bounded continuous function $\phi:\mathcal M_{\rm p} \to \mathbb R$ there holds
\[
 \lim_{n \to \infty} \int_\Omega \phi (\Upsilon_n^\omega) \mathbb{P} (\mathrm{d} \omega) = \int_{\tilde \Omega} \phi (\Upsilon^{\tilde{\omega}}) \tilde{\mathbb{P}} (\mathrm{d} \omega) .
\]
Finally, we introduce a characterization for a region of localization. We refer to \cite{TautenhahnV-13} for a discussion of the validity of Assumption~\ref{ass:FVC}.
Roughly speaking, it is satisfied whenever one of the Theorems \ref{thm:loc1}, \ref{thm:loc2}, or \ref{thm:loc:weak} holds.
\begin{assumption} \label{ass:FVC} 
Let $I \subset \mathbb{R}$. We assume that for all $E \in I$ there exists $\Theta > 3d-1$ such that
\[
 \limsup_{L\to \infty} \mathbb{P} \left\{ \forall x,y \in \Lambda_{L} , \ \lvert x-y \rvert_\infty \geq \frac{L}{2} : \lvert  G_{\omega , \Lambda_{L}} (E;x,y) \rvert \leq L^{-\Theta} \right\} = 1 .
\]
\end{assumption}
Let $\Sigma$ denote the almost sure spectrum of the (ergodic) family of operators $H_\omega$, $\omega\in \Omega$. 
\begin{theorem}[\cite{TautenhahnV-13}] \label{theorem:poisson}
 Let Assumption \ref{ass:circulant} be satisfied, $I \subset \Sigma$ be a bounded interval and $E_0 \in I$. Assume that Assumption \ref{ass:FVC} is satisfied in $I$ and Assumption \ref{ass:pos} is satisfied for $E_0$ and some $\kappa \in [0,1/(1+d))$. 
\par
Then the point process $\Xi$, defined in Eq.~\eqref{eq:process}, converges for $L \to \infty$ weakly to a Poisson process on $\mathbb{R}$ with
Lebesgue measure as the intensity measure.
\end{theorem}
The result of Theorem~\ref{theorem:poisson} follows from Minami's estimate (formulated in Theorem~\ref{thm:minami}) by using an abstract result from \cite{GerminetK-11b}. 
Roughly speaking, the criterion of \cite{GerminetK-11b} states, that for a large class of discrete random Schr\"odinger operators Minami's estimate and a Wegner estimate implies Poisson statistics in any region of localization. 
\section{Reverse H\"older inequality and fractional moments}
\label{s:ESS}
In this section we review a result of Elgart Shamis and Sodin \cite{ElgartSS-12}. They apply the fractional moment method for a large class of discrete alloy-type models. 
The main new ingredient in comparison to other proofs via the fractional moment method is an estimate on the integral of a fractional power of a rational function, 
respectively, an iterated version thereof.  
\par
We would like to point out that the estimate which is effectively used in \cite{ElgartSS-12}, is a reverse H\"older inequality. Such inequalities play an important role
in harmonic analysis, e.g.\ in the theory of Muckenhoupt weights. 
In this section we modify the method of \cite{ElgartSS-12} for the discrete alloy-type model 
at large disorder without the use of the iterated version of the reverse H\"older inequality.
\par
As mentioned before, we assume throughout this section that Assumption~\ref{ass:alloy} is satisfied.
First we state additional regularity assumptions for the model, see  \cite{ElgartSS-12}.
\begin{assumption} \label{ass:hoelder}
 There exists $\alpha \in (0,1]$ and $C_{1} > 0$, such that $\mu ([t-\epsilon , t+ \epsilon]) \leq  C_{1} \epsilon^\alpha$ for all $\epsilon > 0$ and $t \in \RR$.
\end{assumption}
\begin{assumption} \label{ass:moment}
 $\mu$ has a finite $q$-moment, i.e.\ there exists $q > 0$ and a constant $C_{2}>0$ such that $\int \lvert x^q \rvert \mu (\drm x) \leq C_{2}$.
\end{assumption}
\begin{assumption} \label{ass:finite}
 We assume that $\Theta$ is a finite set and that $0 \in \Theta$.
\end{assumption}
The next lemma provides the usual boundedness of fractional moments.
\begin{lemma} \label{lemma:finiteness}
 Let Assumption~\ref{ass:hoelder} be satisfied and $s \in (0,\alpha)$. Then we have for all $b \in \CC$
 \[
  \int_\RR \frac{1}{\lvert x-b \rvert^s} \mu (\drm x) \leq C_{1}^{s/\alpha} \frac{\alpha}{\alpha - s} 
 \]
\end{lemma}
\begin{proof}
We assume $b \in \RR$, if $b \not \in \RR$ we estimate the integrand by replacing $b$ by its real part. Layer Cake gives us
 \[
  I := \int_\RR \frac{1}{\lvert x-b \rvert^s} \mu (\drm x) = \int_0^\infty \mu (\{ x \in \RR \colon \lvert x-b \rvert^{-s} > t \}) \drm t .
 \]
We split the domain of integration for some $\kappa > 0$ according to $[0,\infty) = [0,\kappa) \cup [\kappa , \infty)$ and obtain
\[
  I =\int_0^\kappa \mu (\{ x \in \RR \colon \lvert x-b \rvert^{-s} > t \}) \drm t 
  +\int_\kappa^\infty \mu (\{ x \in \RR \colon \lvert x-b \rvert^{-s} > t \}) \drm t
  .  
\]
Since $\mu$ is a probability measure, we can estimate the first integral by $\kappa$. For the second integral we get due to Assumption~\ref{ass:hoelder}
\begin{align*}
 \int_\kappa^\infty \mu (\{ x \in \RR \colon \lvert x-b \rvert^{-s} > t \}) \drm t &= \int_\kappa^\infty \mu ([b-t^{-1/s} , b+t^{-1/s}]) \drm t \\
   &\leq \int_\kappa^\infty C_{1} t^{-\alpha/s} \drm t \\
   &= C_{1} \frac{s}{\alpha - s} \kappa^{-\alpha / s + 1} .
\end{align*}
If we choose $\kappa = C_{1}^{s/\alpha}$ we obtain the statement of the lemma.
\end{proof}
The main new idea of  \cite{ElgartSS-12}, formulated there in Proposition~3.1, implies the following reverse H\"older inequality.
\begin{proposition}[\cite{ElgartSS-12}] \label{proposition:poly}
Let Assumptions \ref{ass:hoelder} and \ref{ass:moment} be satisfied, $Q_1$ and $Q_2$ be two polynomials of degree smaller or equal $k$, and $s \in (0, q \alpha / \min\{k(4\alpha + q) , \alpha / (2k)\})$. Then there is a constant $\tilde C = \tilde C (\alpha , q , k , s, C_1 , C_2)$ such that
\[
 \left(\int_\RR \frac{\lvert Q_1 (x) \rvert^{2s}}{\lvert Q_2 (x) \rvert^{2s}} \drm \mu (x) \right)^{1/2} \leq \tilde C \int_\RR \frac{\lvert Q_1 (x) \rvert^{s}}{\lvert Q_2 (x) \rvert^{s}} \drm \mu (x)
\]
\end{proposition}
The next statement is contained in \cite{ElgartSS-12} as well. We give a short, direct proof, for discrete alloy-type models, which makes use of Lemma~\ref{lemma:finiteness}
and Proposition~\ref{proposition:poly} only.
\begin{theorem} \label{thm:finite_volume}
 Let Assumption~\ref{ass:hoelder}, \ref{ass:moment} and \ref{ass:finite} be satisfied, $s \in (0, q \alpha / \min\{\lvert \Theta \rvert(4\alpha + q) , \alpha / (2k)\})$ and $\Lambda \subset \ZZ^d$. Then there is a constant $C = C(\alpha , q , \lvert \Theta \rvert , s , C_1 , C_2, u(0))$, such that for all $\lambda > 0$, $E \in \RR$ and $x,y\in \Lambda$ with $x \not = y$
\begin{equation}
 \label{eq:recursion} 
\EE (\lvert G_{\omega , \Lambda} (E;x,y) \rvert^s) \leq \frac{C}{\lambda^s} \sum_{\lvert e \rvert = 1} \EE (\lvert G_{\omega , \Lambda} (E;x,y+e) \rvert^s) .
\end{equation}
 \end{theorem}
Here we use the convention that $G_{\omega , \Lambda} (E;x,y) = 0$ if $x \not \in \Lambda$ or $y \not \in \Lambda$. Moreover, we note that the set of $\omega \in \Omega$ such that $E \in \RR$ is in the spectrum of $H_{\omega , \Lambda}$ has $\PP$-measure zero. This justifies to deal with real energies.
\begin{proof}[Proof of Theorem~\ref{thm:finite_volume}]
 By definition of $G_{\omega , \Lambda} (E)$ we have for $x \not = y$
 \begin{align*}
0 &= \langle \delta_x , G_{\omega , \Lambda} (E) (H_{\omega , \Lambda} - E) \delta_y \rangle  = \sum_{i \in \Lambda} G_{\omega , \Lambda} (E;x,i) \langle \delta_i , (H_{\omega , \Lambda} - E) \delta_y \rangle \\
&= -\sum_{\lvert e \rvert = 1} G_{\omega , \Lambda} (E;x,y+e) + ( \lambda V_\omega (y) - E) G_{\omega , \Lambda} (E;x,y)
 \end{align*}
Hence,
\[
 \lvert V_\omega (y) - E/\lambda \rvert^s \lvert G_{\omega , \Lambda} (E;x,y) \rvert^s \leq \frac{1}{\lambda^s} \sum_{\lvert e \rvert = 1} \lvert G_{\omega , \Lambda} (E;x,y) \rvert^s .
\]
Next we provide a lover bound on the expectation of the left hand side. By Cauchy Schwarz we have for all $k \in \ZZ^d$
\begin{multline*}
 \EE_{\{k\}} \bigl( \lvert G_{\omega , \Lambda} (E;x,y) \rvert^{s/2} \bigr)^2 \\ \leq 
 \EE_{\{k\}} \bigl( \lvert G_{\omega , \Lambda} (E;x,y) \rvert^{s} \lvert V_\omega (y) - E/\lambda \rvert^s \bigr) \EE_{\{k\}} \bigl( \lvert V_\omega (y) - E / \lambda  \rvert^{-s} \bigr) .
\end{multline*}
Here $\EE_{\{k\}}$ denotes the expectation with respect to the random variable $\omega_k$, i.e.\ $\EE_{\{k\}} (\cdot) = \int_\RR (\cdot) \mu (\drm \omega_k)$
By Lemma~\ref{lemma:finiteness} (and since $0 \in \Theta$) we have 
\begin{align*}
 \EE_{\{y\}} \bigl( \lvert V_\omega (y) - E / \lambda  \rvert^{-s} \bigr) &=
 \int_\RR \frac{1}{\lvert V_\omega (y) - E/\lambda \rvert^s} \mu (\drm \omega_y) \leq \frac{1}{\lvert u(0) \rvert^s} C_{\rm A1}^{s/\alpha} \frac{\alpha}{\alpha - s} .
\end{align*}
Hence,
\[
 \EE_{\{y\}} \bigl( \lvert G_{\omega , \Lambda} (E;x,y) \rvert^{s/2} \bigr)^2 \leq  \frac{C_{\rm A1}^{s/\alpha}}{\lvert u(0) \rvert^s}  \frac{\alpha}{\alpha - s}
 \EE_{\{y\}} \bigl( \lvert G_{\omega , \Lambda} (E;x,y) \rvert^{s} \lvert V_\omega (y) - E/\lambda \rvert^s \bigr) .
\]
By Cramers rule, the Green function is a ratio of two polynomials. More precisely, we have
\[
 G_{\omega , \Lambda} (E;x,y) = \frac{\det C_{y,x}}{\det (H_{\omega , \Lambda} - E)} ,
\]
where $C_{i,j} = (-1)^{i+j}M_{i,j}$, and where $M_{i,j}$ is obtained from the matrix $H_{\omega , \Lambda} - E$ by deleting row $i$ and column $j$. Now we observe that both, the numerator and the denominator, are polynomials in $\omega_y$ of order $k \leq \lvert \Theta \rvert$. By Lemma~\ref{proposition:poly} there is a constant $\tilde C = \tilde C (\alpha , q , \lvert \Theta \rvert , s , C_1 , C_2)$ such that
\[
 \EE_{\{y\}} \bigl( \lvert G_{\omega , \Lambda} (E;x,y) \rvert^s \bigr)^2 \geq {\tilde C}^{-1}
 \EE_{\{y\}} \bigl( \lvert G_{\omega , \Lambda} (E;x,y) \rvert^{2s} \bigr)
\]
Hence,
\[
  \EE_{\{y\}} \bigl( \lvert G_{\omega , \Lambda} (E;x,y) \rvert^{s} \bigr) \leq \tilde C \frac{C_{\rm A1}^{s/\alpha}}{\lvert u(0) \rvert^s}  \frac{\alpha}{\alpha - s}
 \EE_{\{y\}} \bigl( \lvert G_{\omega , \Lambda} (E;x,y) \rvert^{s} \lvert V_\omega (y) - E/\lambda \rvert^s \bigr) .
\]
Putting everything together we obtain the statement of the theorem.
\end{proof}
\begin{remark}
 The conclusion of Theorem~\ref{thm:finite_volume} implies the fractional moment bound as in Ineq.~\eqref{eq:fmb}, if $\EE ( \lvert G_{\omega , \Lambda} (E;x,y) \rvert^s)$ is uniformly bounded and $\lambda$ is sufficiently large. 
This is elaborated e.g.\ in \cite{Graf-94}. 
So the question remains, whether 
\[
\sup\limits_{\Lambda \subset \ZZ^d, x,y \in\Lambda}\EE ( \lvert G_{\omega , \Lambda} (E;x,y) \rvert^s)< \infty  .
\]
An elementary argument how to deduce this uniform bound from \eqref{eq:recursion} in the large disorder regime, is given in Corollaries~2.3 and 2.4 of  \cite{ElgartSS-12}.  
Hence, exponential decay and localization follows.
\end{remark}
\newcommand{\etalchar}[1]{$^{#1}$}

\end{document}